\documentstyle[12pt,fleqn]{article}

\font\twlgot =eufm10 scaled \magstep1 \font\egtgot =eufm8
\font\sevgot =eufm7

\font\twlmsb =msbm10 scaled \magstep1 \font\egtmsb =msbm8
\font\sevmsb =msbm7

\newfam\gotfam
\def\pgot{\fam\gotfam\twlgot}
\textfont\gotfam\twlgot \scriptfont\gotfam\egtgot
\scriptscriptfont\gotfam\sevgot
\def\got{\protect\pgot}

\newfam\msbfam
\textfont\msbfam\twlmsb \scriptfont\msbfam\egtmsb
\scriptscriptfont\msbfam\sevmsb
\def\Bbb{\protect\pBbb}
\def\pBbb{\relax\ifmmode\expandafter\Bb\else\typeout{You cann't use
Bbb in text mode}\fi}
\def\Bb #1{{\fam\msbfam\relax#1}}

\def\op#1{\mathop{{\it\fam0} #1}\limits}

\newcommand{\beq}{\begin{equation}}
\newcommand{\eeq}{\end{equation}}
\newcommand{\ben}{\begin{eqnarray}}
\newcommand{\een}{\end{eqnarray}}
\newcommand{\be}{\begin{eqnarray*}}
\newcommand{\ee}{\end{eqnarray*}}

\newcommand{\gG}{{\got G}}

\newcommand{\cA}{{\cal A}}
\newcommand{\cO}{{\cal O}}

\newcommand{\cL}{{\cal L}}

\newcommand{\cE}{{\cal E}}

\newcommand{\al}{\alpha}
\newcommand{\bt}{\beta}
\newcommand{\dl}{\delta}
\newcommand{\la}{\lambda}

\newcommand{\F}{\Phi}
\newcommand{\p}{\pi}

\newcommand{\Om}{\Omega}
\newcommand{\m}{\mu}
\newcommand{\n}{\nu}
\newcommand{\g}{\gamma}

\newcommand{\e}{\epsilon}

\newcommand{\si}{\sigma}

\newcommand{\w}{\wedge}

\newcommand{\wt}{\widetilde}
\newcommand{\wh}{\widehat}
\newcommand{\ol}{\overline}

\newcommand{\dr}{\partial}

\newcommand{\ot}{\otimes}

\newcounter{eqalph}
\newcounter{equationa}
\newcounter{example}
\newcounter{remark}
\newcounter{theorem}
\newcounter{proposition}
\newcounter{lemma}
\newcounter{corollary}
\newcounter{definition}

\def\theremark{\arabic{remark}}

\def\thedefinition{\arabic{definition}}

\newenvironment{proof}{\noindent {\bf Proof.}}{\hfill{\footnotesize\bf
QED} \bigskip }
\newenvironment{example}{\refstepcounter{remark} {\bf Example
\theremark.}}{}
\newenvironment{remark}{\refstepcounter{remark} {\bf Remark
\theremark.}}{}
\newenvironment{theorem}{\refstepcounter{definition} {\sc
Theorem \thedefinition}.}{$\Box$ }

\newenvironment{lemma}{\refstepcounter{definition} {\sc Lemma
\thedefinition}.}{ $\Box$ }

\newcommand{\mar}[1]{}

\begin{document}

\hbox{}

\begin{center}

{\large \bf Relativistic mechanics in a general setting}

\bigskip
\bigskip

{\sc G. Sardanashvily}
\bigskip

Department of Theoretical Physics, Moscow State University,
Moscow, Russia

\end{center}

\bigskip
\bigskip

\begin{small}

\noindent {\bf Abstract}. Relativistic mechanics on an arbitrary
manifold is formulated in the terms of jets of its one-dimensional
submanifolds. A generic relativistic Lagrangian is constructed.
Relativistic mechanics on a pseudo-Riemannian manifold is
particularly considered.
\end{small}


\bigskip
\bigskip

\section{Introduction}

Classical non-relativistic mechanics is adequately formulated as
Lagrangian and Hamiltonian theory on a fibre bundle $Q\to\Bbb R$
over the time axis $\Bbb R$
\cite{eche,book05,epr,leon,book98,sard98}.

If a configuration space $Q$ of a mechanical system has no
preferable fibration $Q\to\Bbb R$, we obtain a general formulation
of relativistic mechanics, including Special Relativity on the
Minkowski space $Q=\Bbb R^4$ \cite{book05,book98,sard98}. A
velocity space of relativistic mechanics is the first order jet
manifold $J^1_1Q$ of one-dimensional submanifolds of the
configuration space $Q$. The notion of jets of submanifolds
\cite{book,epr2,book09,kras} generalizes that of jets of sections
of fibre bundles which are utilized in field theory and
non-relativistic mechanics (Section 2). The jet bundle $J^1_1Q\to
Q$ is projective, and one can think of its fibres as being spaces
of the three-velocities of a relativistic system (Section 3). The
four-velocities of a relativistic system are represented by
elements of the tangent bundle $TQ$ of the configuration space
$Q$, while the cotangent bundle $T^*Q$, endowed with the canonical
symplectic form, plays a role of the phase space of relativistic
theory (Section 6).

We develop Lagrangian formalism on the jet bundle $J^1_1Q\to Q$
(Section 4). We show that, in the framework of this formalism,
Lagrangians possess a certain gauge symmetry (\ref{kk33}) and,
consequently, the corresponding Lagrange operators obey the rather
restrictive Noether identity (\ref{s60}). Solving this Noether
identity, we obtain the generic Lagrangian (\ref{kk1}) and the
equation of motion (\ref{k34}) of relativistic mechanics on a
manifold $Q$. In particular, if $Q$ is the Minkowski space, we are
in the case of Special Relativity (Example \ref{0313}).

Generalizing this example, we consider relativistic mechanics on
an arbitrary pseudo-Riemannian manifold. Its equation of motion is
the relativistic geodesic equation (\ref{cqg2}). Hamiltonian
relativistic mechanics on a pseudo-Riemannian manifold is
developed in Section 6. Its generic Hamiltonian takes the form
(\ref{kk101}).

\section{Jets of submanifolds}

Jets of sections of fibre bundles are particular jets of
submanifolds of a manifold \cite{book,epr2,book09,kras}.

Given an $m$-dimensional smooth real manifold $Z$, a $k$-order jet
of $n$-dimensional submanifolds of $Z$ at a point $z\in Z$ is
defined as an equivalence class $j^k_zS$ of $n$-dimensional
imbedded submanifolds of $Z$ through $z$ which are tangent to each
other at $z$ with order $k\geq 0$. Namely, two submanifolds
\be
i_S: S\to Z,\qquad  i_{S'}: S'\to Z
\ee
through a point $z\in Z$ belong to the same equivalence class
$j^k_zS$ if and only if the images of the $k$-tangent morphisms
\be
T^ki_S: T^kS\to T^kZ, \qquad  T^ki_{S'}: T^kS'\to T^kZ
\ee
coincide with each other. The set
\be
J^k_nZ=\op\bigcup_{z\in Z} j^k_zS
\ee
of $k$-order jets of submanifolds is a finite-dimensional real
smooth manifold, called the $k$-order jet manifold of
submanifolds. For the sake of convenience, we put $J^0_nZ =Z$.

If $k>0$, let $Y\to X$ be an $m$-dimensional fibre bundle over an
$n$-dimensional base $X$ and $J^kY$ the $k$-order jet manifold of
sections of $Y\to X$. Given an imbedding $\Phi:Y\to Z$, there is
the natural injection
\mar{kk83}\beq
J^k\Phi: J^kY\to J^k_nZ, \qquad j^k_xs \to [\Phi\circ
s]^k_{\Phi(s(x))}, \label{kk83}
\eeq
where $s$ are sections of $Y\to X$. This injection defines a chart
on $J^k_nZ$. These charts provide a manifold atlas of $J^k_nZ$.

Let us restrict our consideration to first order jets of
submanifolds. There is obvious one-to-one correspondence
\mar{z793}\beq
\la_{(1)}: j^1_zS \to V_{j^1_zS}\subset T_zZ  \label{z793}
\eeq
between the jets $j^1_zS$ at a point $z\in Z$ and the
$n$-dimensional vector subspaces of the tangent space $T_zZ$ of
$Z$ at $z$. It follows that $J^1_nZ$ is a fibre bundle
\mar{s3}\beq
\rho:J^1_nZ\to Z \label{s3}
\eeq
with the structure group $GL(n,m-n;\Bbb R)$ of linear
transformations of the vector space $\Bbb R^m$ which preserve its
subspace $\Bbb R^n$. The typical fibre of the fibre bundle
(\ref{s3}) is the Grassmann manifold
\be
\gG(n,m-n;\Bbb R)=GL(m;\Bbb R)/GL(n,m-n;\Bbb R).
\ee
This fibre bundle possesses the following coordinate atlas.

Let $\{(U;z^A)\}$ be a coordinate atlas of $Z$. Though $J^0_nZ=Z$,
let us provide $J^0_nZ$ with an atlas where every chart $(U;z^A)$
on a domain $U\subset Z$ is replaced with the
\be
{m\choose n}=\frac{m!}{n!(m-n)!}
\ee
charts on the same domain $U$ which correspond to different
partitions of the collection $(z^1\cdots z^A)$ in the collections
of $n$ and $m-n$ coordinates
\mar{5.8}\beq
(U; x^\la,y^i), \qquad \la=1,\ldots,n,  \qquad
i=1,\ldots,m-n.\label{5.8}
\eeq
The transition functions between the coordinate charts (\ref{5.8})
of $J^0_nZ$ associated with a coordinate chart $(U,z^A)$ of $Z$
are reduced to exchange between coordinates $x^\la$ and $y^i$.
Transition functions between arbitrary coordinate charts of the
manifold $J^0_nZ$  take the form
\mar{5.26} \beq
x'^\la = x'^\la (x^\m, y^k), \qquad y'^i = y'^i (x^\m, y^k).
\label{5.26}
\eeq

Given the coordinate atlas (\ref{5.8}) -- (\ref{5.26}) of a
manifold $J^0_nZ$, the first order jet manifold $J^1_nZ$  is
endowed with an atlas of adapted coordinates
\mar{5.31}\beq
(\rho^{-1}(U)=U\times\Bbb R^{(m-n)n}; x^\la,y^i,y^i_\la),
\label{5.31}
\eeq
possessing transition functions
\mar{5.36}\beq
y'^i_\la =\left(\frac{\dr y'^i}{\dr y^j}y^j_\al +\frac{\dr
y'^i}{\dr x^\al}\right) \left(\frac{\dr x^\al}{\dr y'^k}y'^k_\la
+\frac{\dr x^\al}{\dr x'^\la}\right). \label{5.36}
\eeq

\section{Relativistic mechanics}

As was mentioned above, a velocity space of relativistic mechanics
is the first order jet manifold $J^1_1Q$ of one-dimensional
submanifolds of a configuration space $Q$
\cite{book05,book98,sard98}.

Given an $m$-dimensional manifold $Q$ coordinated by $(q^\la)$,
let us consider the jet manifold $J^1_1Q$ of its one-dimensional
submanifolds. Let us provide $Q=J^0_1Q$ with the coordinates
(\ref{5.8}):
\mar{0303}\beq
(U;x^0=q^0, y^i=q^i)= (U;q^\la). \label{0303}
\eeq
Then the jet manifold
\be
\rho:J^1_1Q\to Q
\ee
is endowed with coordinates (\ref{5.31}):
\mar{0300}\beq
(\rho^{-1}(U);q^0,q^i,q^i_0) \label{0300}
\eeq
possessing transition functions (\ref{5.26}), (\ref{5.36}) which
read
\mar{s120,'}\ben
&&q'^0=q'^0(q^0,q^k), \qquad q'^0=q'^0(q^0,q^k), \label{s120}\\
&&q'^i_0= (\frac{\dr q'^i}{\dr q^j} q^j_0 + \frac{\dr q'^i}{\dr
q^0} ) (\frac{\dr q'^0}{\dr q^j} q^j_0 + \frac{\dr q'^0}{\dr q^0}
)^{-1}. \label{s120'}
\een
A glance at the transformation law (\ref{s120'}) shows that
$J^1_1Q\to Q$ is a fibre bundle in projective spaces.

\begin{example} \label{0310} \mar{0310}
Let $Q=M^4=\Bbb R^4$ be a Minkowski space whose Cartesian
coordinates $(q^\la)$, $\la=0,1,2,3,$ are subject to the Lorentz
transformations (\ref{s120}):
\mar{s122}\beq
q'^0= q^0{\rm ch}\al - q^1{\rm sh}\al, \quad q'^1= -q^0{\rm sh}\al
+ q^1{\rm ch}\al, \quad q'^{2,3} = q^{2,3}. \label{s122}
\eeq
Then $q'^i$ (\ref{s120'}) are exactly the Lorentz transformations
\be
q'^1_0=\frac{ q^1_0{\rm ch}\al -{\rm sh}\al}{ - q^1_0{\rm sh}\al+
{\rm ch}\al} \qquad q'^{2,3}_0=\frac{q^{2,3}_0}{ - q^1_0{\rm
sh}\al + {\rm ch}\al}
\ee
of three-velocities in relativistic mechanics
\cite{book98,sard98}.
\end{example}

In view of Example \ref{0310}, one can think of the velocity space
$J^1_1Q$ of relativistic mechanics as being a space of
three-velocities. For the sake of convenience, we agree to call
$J^1_1Q$ the three-velocity space and its coordinate
transformations (\ref{s120}) -- (\ref{s120'}) the relativistic
transformations, though a dimension of $Q$ need not equal $3+1$.

\section{Lagrangian relativistic mechanics}

Given the coordinate chart (\ref{0300}) of $J^1_1Q$, one can
regard $\rho^{-1}(U)\subset J^1_1Q$ as the first order jet
manifold $J^1U$ of sections of the fibre bundle
\mar{0301}\beq
\pi:U\ni (q^0,q^i)\to (q^0)\in \pi(U)\subset \Bbb R. \label{0301}
\eeq
Then three-velocities $(q^i_0)\in \rho^{-1}(U)$ of a relativistic
system on $U$ can be treated as absolute velocities of a local
non-relativistic system on the configuration space $U$
(\ref{0301}). However, this treatment is broken under the
relativistic transformations $q^i_0\to q'^i_0$ (\ref{s120}) since
they are not affine. One can develop first order Lagrangian
formalism with a Lagrangian
\be
L=\cL dq^0\in \cO^{0,1}(\rho^{-1}(U))
\ee
on a coordinate chart $\rho^{-1}(U)$, but this Lagrangian fails to
be globally defined on $J^1_1Q$ (see Remark \ref{kk50} below). The
graded differential algebra $\cO^*(\rho^{-1}(U))$ of exterior
forms on $\rho^{-1}(U)$ is generated by horizontal forms $dq^0$
and contact forms $dq^i-q^i_0 dq^0$. Coordinate transformations
(\ref{s120}) preserve the ideal of contact forms, but horizontal
forms are not transformed into horizontal forms, unless coordinate
transition functions $q^0$ (\ref{s120}) are independent of
coordinates $q'^i$.

In order to overcome this difficulty, let us consider a trivial
fibre bundle
\mar{str}\beq
Q_R=\Bbb R\times Q\to \Bbb R, \label{str}
\eeq
whose base $\Bbb R$ is endowed with a Cartesian coordinate $\tau$
\cite{book09}. This fibre bundle is provided with an atlas of
coordinate charts
\mar{s20'}\beq
(\Bbb R\times U; \tau,q^\la), \label{s20'}
\eeq
where $(U; q^0,q^i)$ are the coordinate charts (\ref{0303}) of the
manifold $J^0_1Q$. The coordinate charts (\ref{s20'}) possess
transition functions (\ref{s120}). Let $J^1Q_R$ be the first order
jet manifold of the fibre bundle (\ref{str}). Since the
trivialization (\ref{str}) is fixed, there is the canonical
isomorphism of $J^1Q_R$ to the vertical tangent bundle
\mar{s30}\beq
J^1Q_R= VQ_R= \Bbb R\times TQ \label{s30}
\eeq
of $Q_R\to \Bbb R$ \cite{book09,epr}.

Given the coordinate atlas (\ref{s20'}) of $Q_R$, the jet manifold
$J^1Q_R$ is endowed with the coordinate charts
\mar{s14}\beq
 ((\pi^1)^{-1}(\Bbb R\times U)=\Bbb R\times U\times\Bbb R^m;
\tau,q^\la,q^\la_\tau), \label{s14}
\eeq
possessing transition functions
\mar{s16'}\beq
q'^\la_\tau=\frac{\dr q'^\la}{\dr q^\m}q^\m_\tau. \label{s16'}
\eeq
Relative to the coordinates (\ref{s14}), the isomorphism
(\ref{s30}) takes the form
\mar{0305}\beq
(\tau,q^\m,q^\m_\tau) \to (\tau,q^\m,\dot q^\m=q^\m_\tau).
\label{0305}
\eeq

\begin{example} \label{0311} \mar{0311} Let $Q=M^4$ be a Minkowski
space in Example \ref{0310} whose Cartesian coordinates
$(q^0,q^i)$ are subject to the Lorentz transformations
(\ref{s122}). Then the corresponding transformations (\ref{s16'})
take the form
\be
q'^0_\tau= q^0_\tau{\rm ch}\al - q^1_\tau{\rm sh}\al, \quad
q'^1_\tau= -q^0_\tau{\rm sh}\al + q^1_\tau{\rm ch}\al, \quad
q'^{2,3}_\tau = q^{2,3}_\tau
\ee
of transformations of four-velocities in relativistic mechanics.
\end{example}

In view of Example \ref{0311}, we agree to call fibre elements of
$J^1Q_R\to Q_R$ the four-velocities though the dimension of $Q$
need not equal 4. Due to the canonical isomorphism $q^\la_\tau\to
\dot q^\la$ (\ref{s30}), by four-velocities also are meant the
elements of the tangent bundle $TQ$, which is called the space of
four-velocities.

Obviously, the non-zero jet (\ref{0305}) of sections of the fibre
bundle (\ref{str}) defines some jet of one-dimensional subbundles
of the manifold $\{\tau\}\times Q$ through a point $(q^0,q^i)\in
Q$, but this is not one-to-one correspondence.

Since non-zero elements of $J^1Q_R$ characterize jets of
one-dimensional submanifolds of $Q$, one hopes to describe the
dynamics of one-dimensional submanifolds of a manifold $Q$ as that
of sections of the fibre bundle (\ref{str}). For this purpose, let
us refine the relation between elements of the jet manifolds
$J^1_1Q$ and $J^1Q_R$.

Let us consider the manifold product $\Bbb R\times J^1_1Q$. It is
a fibre bundle over $Q_R$.  Given a coordinate atlas (\ref{s20'})
of $Q_R$, this product is endowed with the coordinate charts
\mar{s13}\beq
(U_R\times \rho^{-1}(U)=U_R\times U\times\Bbb R^{m-1};
\tau,q^0,q^i, q^i_0), \label{s13}
\eeq
possessing transition functions (\ref{s120}) -- (\ref{s120'}). Let
us assign to an element $(\tau,q^0,q^i, q^i_0)$ of the chart
(\ref{s13}) the elements $(\tau,q^0,q^i,q^0_\tau, q^i_\tau)$ of
the chart (\ref{s14}) whose coordinates obey the relations
\mar{s17}\beq
q^i_0 q^0_\tau = q^i_\tau. \label{s17}
\eeq
These elements make up a one-dimensional vector space. The
relations (\ref{s17}) are maintained under  coordinate
transformations (\ref{s120'}) and (\ref{s16'}) \cite{epr2,book09}.
Thus, one can associate:
\mar{s25}\beq
(\tau,q^0,q^i, q^i_0) \to \{(\tau,q^0,q^i,q^0_\tau, q^i_\tau) \, |
\, q^i_0 q^0_\tau = q^i_\tau\}, \label{s25}
\eeq
to each element of the manifold $\Bbb R\times J^1_1Q$ a
one-dimensional vector space in the jet manifold $J^1Q_R$. This is
a subspace of elements
\be
q^0_\tau (\dr_0 + q^i_0\dr_i)
\ee
of a fibre of the vertical tangent bundle (\ref{s30}) at a point
$(\tau,q^0,q^i)$. Conversely, given a non-zero element
(\ref{0305}) of $J^1Q_R$, there is a coordinate chart (\ref{s14})
such that this element defines a unique element of $\Bbb R\times
J^1_1Q$ by the relations
\mar{s31}\beq
q^i_0=\frac{q^i_\tau}{q^0_\tau}. \label{s31}
\eeq

Thus, we have shown the following. Let $(\tau,q^\la)$ further be
arbitrary coordinates on the product $Q_R$ (\ref{str}) and
$(\tau,q^\la,q^\la_\tau)$ the corresponding coordinates on the jet
manifold $J^1Q_R$.

\begin{theorem} \label{s50} \mar{s50}
(i) Any jet of submanifolds through a point $q\in Q$ defines some
(but not unique) jet of sections of the fibre bundle $Q_R$
(\ref{str}) through a point $\tau\times q$ for any $\tau\in \Bbb
R$ in accordance with the relations (\ref{s17}).

(ii)  Any non-zero element of $J^1Q_R$ defines a unique element of
the jet manifold $J^1_1Q$ by means of the relations (\ref{s31}).
However, non-zero elements of $J^1Q_R$ can correspond to different
jets of submanifolds.

(iii) Two elements $(\tau,q^\la,q^\la_\tau)$ and
$(\tau,q^\la,q'^\la_\tau)$ of $J^1Q_R$ correspond to the same jet
of submanifolds if $q'^\la_\tau=r q^\la_\tau$, $r\in\Bbb
R\setminus \{0\}$.
\end{theorem}

In the case of a Minkowski space $Q=M^4$ in Examples \ref{0310}
and \ref{0311}, the equalities (\ref{s17}) and (\ref{s31}) are the
familiar relations between three- and four-velocities.

Based on Theorem \ref{s50}, we can develop Lagrangian theory of
one-dimensional submanifolds of a manifold $Q$ as that of sections
of the fibre bundle $Q_R$ (\ref{str}). Let
\mar{s40}\beq
L=\cL(\tau,q^\la, q^\la_\tau) d\tau, \label{s40}
\eeq
be a first order Lagrangian on the jet manifold $J^1Q_R$. The
corresponding Lagrange operator reads
\mar{s41}\beq
\dl L= \cE_\la dq^\la\w d\tau, \qquad \cE_\la= \dr_\la\cL - d_\tau
\dr_\la^\tau\cL. \label{s41}
\eeq
It yields the Lagrange equation
\mar{s90}\beq
\cE_\la= \dr_\la\cL - d_\tau \dr_\la^\tau\cL =0. \label{s90}
\eeq

In accordance with Theorem \ref{s50}, it seems reasonable to
require that, in order to describe jets of one-dimensional
submanifolds of $Q$, the Lagrangian $L$ (\ref{s40}) on $J^1Q_R$
possesses a gauge symmetry given by vector fields
$u=\chi(\tau)\dr_\tau$ on $Q_R$ or, equivalently, their vertical
part
\mar{kk33}\beq
u_V= - \chi q^\la_\tau\dr_\la, \label{kk33}
\eeq
which are generalized vector fields on $Q_R$ \cite{book09,epr}.
Then the variational derivatives of this Lagrangian obey the
Noether identity:
\mar{s60}\beq
q^\la_\tau\cE_\la=0. \label{s60}
\eeq
We call such a Lagrangian the relativistic Lagrangian.

In order to obtain a generic form of a relativistic Lagrangian
$L$, let us regard the Noether identity (\ref{s60}) as an equation
for $L$. It admits the following solution. Let
\be
\frac1{2N!}G_{\al_1\ldots\al_{2N}}(q^\nu)dq^{\al_1}\vee\cdots\vee
dq^{\al_{2N}}
\ee
be a symmetric tensor field on $Q$ such that the function
\mar{kk6}\beq
G=G_{\al_1\ldots\al_{2N}}(q^\nu)\dot q^{\al_1}\cdots\dot
q^{\al_{2N}} \label{kk6}
\eeq
is positive:
\mar{kk7}\beq
G>0, \label{kk7}
\eeq
everywhere on $TQ\setminus \wh 0(Q)$. Let $A=A_\m(q^\nu) dq^\m$ be
a one-form on $Q$. Given the pull-back of $G$ and $A$ onto
$J^1Q_R$ due to the canonical isomorphism (\ref{s30}), we define a
Lagrangian
\mar{kk1}\beq
L=(G^{1/2N} + q^\m_\tau A_\m)d\tau, \qquad
G=G_{\al_1\ldots\al_{2N}}q^{\al_1}_\tau\cdots q^{\al_{2N}}_\tau,
\label{kk1}
\eeq
on $J^1Q_R\setminus (\Bbb R\times \wh 0(Q))$ where $\wh 0$ is the
global zero section of $TQ\to Q$. The corresponding Lagrange
equation reads
\mar{kk2,3}\ben
&&\cE_\la = \frac{\dr_\la G}{2NG^{1-1/2N}} - d_\tau\left(\frac{\dr_\la^\tau
G}{2NG^{1-1/2N}}\right)+ F_{\la\m}q^\m_\tau= \label{kk2}\\
&& \qquad E_\bt[\dl^\bt_\la -
 q^\bt_\tau G_{\la\nu_2\ldots\nu_{2N}}q^{\nu_2}_\tau\cdots
q^{\nu_{2N}}_\tau G^{-1}]G^{1/2N-1}=0, \nonumber\\
&& E_\bt= \left(\frac{\dr_\bt
G_{\m\al_2\ldots\al_{2N}}}{2N}- \dr_\m
G_{\bt\al_2\ldots\al_{2N}}\right) q^\m_\tau q^{\al_2}_\tau\cdots
q^{\al_{2N}}_\tau - \label{kk3}\\
&& \qquad (2N-1)G_{\bt\m\al_3\ldots\al_{2N}}q^\m_{\tau\tau}q^{\al_3}_\tau\cdots
q^{\al_{2N}}_\tau  + G^{1-1/2N}F_{\bt\m}q^\m_\tau,\nonumber\\
&& F_{\la\m}=\dr_\la A_\m-\dr_\m A_\la. \nonumber
\een
It is readily observed that the variational derivatives $\cE_\la$
(\ref{kk2}) satisfy the Noether identity (\ref{s60}). Moreover,
any relativistic Lagrangian obeying the Noether identity
(\ref{s60}) is of type (\ref{kk1}).

A glance at the Lagrange equation (\ref{kk2}) shows that it holds
if
\mar{kk5}\beq
E_\bt= \F G_{\bt\nu_2\ldots\nu_{2N}}q^{\nu_2}_\tau\cdots
q^{\nu_{2N}}_\tau G^{-1}, \label{kk5}
\eeq
where $\F$ is some function on $J^1Q_R$. In particular, we
consider the equation
\mar{kk4}\beq
E_\bt=0. \label{kk4}
\eeq

Because of the Noether identity (\ref{s60}), the system of
equations (\ref{kk2}) is underdetermined. To overcome this
difficulty, one can complete it with some additional equation.
Given the function $G$ (\ref{kk1}), let us choose the condition
\mar{kk8}\beq
G=1. \label{kk8}
\eeq
Owing to the property (\ref{kk7}), the function $G$ (\ref{kk1})
possesses a nowhere vanishing differential. Therefore, its level
surface $W_G$ defined by the condition (\ref{kk8}) is a
submanifold of $J^1Q_R$.

Our choice of the equation (\ref{kk4}) and the condition
(\ref{kk8}) is motivated by the following facts.

\begin{lemma} \label{kk12} \mar{kk12} Any solution of the
Lagrange equation (\ref{kk2}) living in the submanifold $W_G$ is a
solution of the equation (\ref{kk4}).
\end{lemma}

\begin{proof}
A solution of the Lagrange equation (\ref{kk2}) living in the
submanifold $W_G$ obeys the system of equations
\mar{kk13}\beq
\cE_\la=0, \qquad G=1. \label{kk13}
\eeq
Therefore, it satisfies the equality
\mar{kk16}\beq
d_\tau G=0. \label{kk16}
\eeq
Then a glance at the expression (\ref{kk2}) shows that the
equations (\ref{kk13}) are equivalent to the equations
\mar{kk14}\ben
&& E_\la= \left(\frac{\dr_\la
G_{\m\al_2\ldots\al_{2N}}}{2N}- \dr_\m
G_{\la\al_2\ldots\al_{2N}}\right) q^\m_\tau q^{\al_2}_\tau\cdots
q^{\al_{2N}}_\tau - \nonumber\\
&& \qquad (2N-1)G_{\bt\m\al_3\ldots\al_{2N}}q^\m_{\tau\tau}q^{\al_3}_\tau\cdots
q^{\al_{2N}}_\tau  + F_{\bt\m}q^\m_\tau =0, \label{kk14}\\
&& G=G_{\al_1\ldots\al_{2N}}q^{\al_1}_\tau\cdots q^{\al_{2N}}_\tau=1.
\nonumber
\een
\end{proof}

\begin{lemma} \label{kk12'} \mar{kk12'}
Solutions of the equation (\ref{kk4}) do not leave the submanifold
$W_G$ (\ref{kk8}).
\end{lemma}

\begin{proof}
Since
\be
d_\tau G= -\frac{2N}{2N-1}q^\bt_\tau E_\bt,
\ee
any solution of the equation (\ref{kk4}) intersecting the
submanifold $W_G$ (\ref{kk8}) obeys the equality (\ref{kk16}) and,
consequently, lives in $W_G$.
\end{proof}

The system of equations (\ref{kk14}) is called the relativistic
equation. Its components $E_\la$ (\ref{kk3}) are not independent,
but obeys the relation
\be
q^\bt_\tau E_\bt=- \frac{2N-1}{2N}d_\tau G=0, \qquad G=1,
\ee
similar to the Noether identity (\ref{s60}). The condition
(\ref{kk8}) is called the relativistic constraint.

Though the equation (\ref{kk2}) for sections of a fibre bundle
$Q_R\to\Bbb R$ is underdetermined, it is determined if, given a
coordinate chart $(U;q^0,q^i)$ (\ref{0303}) of $Q$ and the
corresponding coordinate chart (\ref{s20'}) of $Q_R$, we rewrite
it in the terms of three-velocities $q^i_0$ (\ref{s31}) as an
equation for sections of a fibre bundle $U\to\pi(U)$ (\ref{0301}).

Let us denote
\mar{kk60}\beq
\ol G(q^\la,q^i_0)=(q^0_\tau)^{-2N}G(q^\la,q^\la_\tau), \qquad
q^0_\tau\neq 0. \label{kk60}
\eeq
Then we have
\be
\cE_i=q^0_\tau\left[\frac{\dr_i \ol G}{2N\ol G^{1-1/2N}} -
(q^0_\tau)^{-1}d_\tau\left(\frac{\dr_i^0 \ol G}{2N\ol
G^{1-1/2N}}\right)+ F_{ij}q^j_0 + F_{i0}\right].
\ee
Let us consider a solution $\{s^\la(\tau)\}$ of the equation
(\ref{kk2}) such that $\dr_\tau s^0$ does not vanish and there
exists an inverse function $\tau(q^0)$. Then this solution can be
represented by sections
\mar{kk51}\beq
s^i(\tau)=(\ol s^i\circ s^0)(\tau) \label{kk51}
\eeq
of the composite bundle
\be
\Bbb R\times U\to \Bbb R\times \p(U)\to \Bbb R
\ee
where $\ol s^i(q^0)=s^i(\tau(q^0))$ are sections of $U\to \pi(U)$
and $s^0(\tau)$ are sections of $\Bbb R\times \p(U)\to \Bbb R$.
Restricted to such solutions, the equation (\ref{kk2}) is
equivalent to the equation
\mar{kk20}\ben
&& \ol\cE_i=\frac{\dr_i \ol G}{2N\ol G^{1-1/2N}} -
d_0\left(\frac{\dr_i^0 \ol G}{2N\ol G^{1-1/2N}}\right)+
\label{kk20} \\
&& \qquad F_{ij}q^j_0 + F_{i0}=0, \nonumber\\
&& \ol\cE_0=-q^i_0\ol\cE_i. \nonumber
\een
for sections $\ol s^i(q^0)$ of a fibre bundle $U\to\pi(U)$.

It is readily observed that the equation (\ref{kk20}) is the
Lagrange equation of the Lagrangian
\mar{kk21}\beq
\ol L=(\ol G^{1/2N} + q^i_0 A_i + A_0)dq^0 \label{kk21}
\eeq
on the jet manifold $J^1U$ of a fibre bundle $U\to\pi(U)$.

\begin{remark} \label{kk50} \mar{kk50}
Both the equation (\ref{kk20}) and the Lagrangian (\ref{kk21}) are
defined only on a coordinate chart (\ref{0303}) of $Q$ since they
are not maintained by transition functions (\ref{s120}) --
(\ref{s120'}).
\end{remark}

A solution $\ol s^i(q^0)$ of the equation (\ref{kk20}) defines a
solution $s^\la(\tau)$ (\ref{kk51}) of the equation (\ref{kk2}) up
to an arbitrary function $s^0(\tau)$. The relativistic constraint
(\ref{kk8}) enables one to overcome this ambiguity as follows.

Let us assume that, restricted to the coordinate chart
$(U;q^0,q^i)$ (\ref{0303}) of $Q$, the relativistic constraint
(\ref{kk8}) has no solution $q^0_\tau=0$. Then it is brought into
the form
\mar{kk62}\beq
(q^0_\tau)^{2N}\ol G(q^\la,q^i_0)=1, \label{kk62}
\eeq
where $\ol G$ is the function (\ref{kk60}). With the condition
(\ref{kk62}), every three-velocity $(q^i_0)$ defines a unique pair
of four-velocities
\mar{kk63}\beq
q^0_\tau = \pm (\ol G(q^\la,q^i_0))^{1/2N}, \qquad
q^i_\tau=q_\tau^0q^i_0. \label{kk63}
\eeq
Accordingly, any solution $\ol s^i(q^0)$ of the equation
(\ref{kk20}) leads to solutions
\be
\tau(q^0)=\pm\int (\ol G(q^0,\ol s^i(q^0),\dr_0\ol
s^i(q_0))^{-1/2N}dq^0, \quad s^i(\tau)=s^0(\tau)(\dr_i\ol
s^i)(s^0(\tau))
\ee
of the equation (\ref{kk13}) and, equivalently, the relativistic
equation (\ref{kk14}).

\begin{example} \label{0313} \mar{0313}
Let $Q=M^4$ be a Minkowski space provided with the Minkowski
metric $\eta_{\m\nu}$ of signature $(+,---)$. This is the case of
Special Relativity. Let $\cA_\la dq^\la$ be a one-form on $Q$.
Then
\mar{s133}\beq
L=[m(\eta_{\m\nu}q^\m_\tau q^\nu_\tau)^{1/2} +e \cA_\m
q^\m_\tau]d\tau, \qquad m,e\in\Bbb R,\label{s133}
\eeq
is a relativistic Lagrangian on $J^1Q_R$ which satisfies the
Noether identity (\ref{s60}). The corresponding relativistic
equation (\ref{kk14}) reads
\mar{x4,'}\ben
&& m\eta_{\mu\nu}q^\nu_{\tau\tau} -eF_{\m\nu}q^\nu_\tau=0,
\label{x4}\\
&& \eta_{\m\nu}q^\m_\tau q^\nu_\tau=1.  \label{x4'}
\een
This describes a relativistic massive charge in the presence of an
electromagnetic field $\cA$. It follows from the relativistic
constraint (\ref{x4'}) that $(q^0_\tau)^2\geq 1$. Therefore,
passing to three-velocities, we obtain the Lagrangian
(\ref{kk21}):
\be
\ol L=\left[m(1-\op\sum_i (q^i_0)^2)^{1/2} +e (\cA_i q^i_0
+\cA_0)\right]dq^0,
\ee
and the Lagrange equation (\ref{kk20}):
\be
d_0\left(\frac{mq^i_0}{(1-\op\sum_i (q^i_0)^2)^{1/2}}\right)
+e(F_{ij}q^j_0 + F_{i0})=0.
\ee
\end{example}

\begin{example} \label{kk65} \mar{kk65}
Let $Q=\Bbb R^4$ be an Euclidean space provided with the Euclidean
metric $\e$. This is the case of Euclidean Special Relativity. Let
$\cA_\la dq^\la$ be a one-form on $Q$. Then
\be
L=[(\e_{\m\nu}q^\m_\tau q^\nu_\tau)^{1/2} +\cA_\m q^\m_\tau]d\tau
\ee
is a relativistic Lagrangian on $J^1Q_R$ which satisfies the
Noether identity (\ref{s60}). The corresponding relativistic
equation (\ref{kk14}) reads
\mar{kk66,'}\ben
&& m \e_{\mu\nu}q^\nu_{\tau\tau} -eF_{\m\nu}q^\nu_\tau=0,
\label{kk66}\\
&& \e_{\m\nu}q^\m_\tau q^\nu_\tau=1.  \label{kk66'}
\een
It follows from the relativistic constraint (\ref{kk66'}) that
$0\leq (q^0_\tau)^2\leq 1$. Passing to three-velocities, one
therefore meets a problem.
\end{example}

\section{Relativistic geodesic equations}

A glance at the relativistic Lagrangian (\ref{kk1}) shows that,
because of the gauge symmetry (\ref{kk33}), this Lagrangian is
independent of $\tau$ and, therefore, it describes an autonomous
mechanical system. Accordingly, the relativistic equation
(\ref{kk14}) on $Q_R$ is conservative and, therefore, it is
equivalent to an autonomous second order equation on $Q$ whose
solutions are parameterized by the coordinate $\tau$ on a base
$\Bbb R$ of $Q_R$. Given holonomic coordinates $(q^\la,\dot q^\la,
\ddot q^\la)$ of the second tangent bundle $T^2Q$, this autonomous
second order equation (called the autonomous relativistic
equation) reads
\mar{k34}\ben
&&\left(\frac{\dr_\la G_{\m\al_2\ldots\al_{2N}}}{2N}- \dr_\m
G_{\la\al_2\ldots\al_{2N}}\right) \dot q^\m \dot q^{\al_2}\cdots
\dot q^{\al_{2N}} - \nonumber\\
&& \qquad (2N-1)G_{\bt\m\al_3\ldots\al_{2N}}\ddot q^\m\dot q^{\al_3}\cdots
\dot q^{\al_{2N}}  + F_{\bt\m}\dot q^\m =0, \label{k34}\\
&& G=G_{\al_1\ldots\al_{2N}}\dot q^{\al_1}\cdots \dot q^{\al_{2N}}=1.
\nonumber
\een
Due to the canonical isomorphism $q^\la_\tau\to \dot q^\la$
(\ref{s30}), the tangent bundle $TQ$ is regarded as a space of
four-velocities.

Generalizing Example \ref{0313}, let us investigate relativistic
mechanics on a four-dimensional pseudo-Riemannian manifold $Q=X$,
coordinated by $(x^\la)$ and provided with a pseudo-Riemannian
metric $g$ of signature $(+,---)$. We agree to call $X$ a world
manifold. Let $A=A_\la dx^\la$ be a one-form on $X$. Let us
consider the relativistic Lagrangian (\ref{kk1}):
\be
L=[(g_{\al\bt}x^\al_\tau x^\bt_\tau)^{1/2} + A_\m x^\m_\tau]d\tau,
\ee
and the relativistic constraint (\ref{kk8}):
\be
g_{\al\bt}x^\al_\tau x^\bt_\tau=1.
\ee
The corresponding autonomous relativistic equation (\ref{kk14}) on
$X$ takes the form
\mar{kk72,1}\ben
&& \ddot x^\la -\{_\m{}^\la{}_\nu\}\dot x^\m\dot x^\nu
-g^{\la\bt}F_{\bt\nu}\dot x^\nu=0, \label{kk72}\\
&& g=g_{\al\bt}\dot x^\al \dot x^\bt=1, \label{kk71}
\een
where $\{_\m{}^\la{}_\nu\}$ is the Levi--Civita connection. A
glance at the equality (\ref{kk72}) shows that it is a geodesic
equation on $TX$ with respect to an affine connection
\mar{kk75}\beq
K_\m^\la= \{_\m{}^\la{}_\nu\} \dot x^\nu + g^{\la\nu}F_{\nu\m}.
\label{kk75}
\eeq
on $TX$.

A particular form of this connection follows from the fact that
the geodesic equation (\ref{kk72}) is derived from a Lagrange
equation, i.e., we are in the case of Lagrangian relativistic
mechanics. In a general setting, relativistic mechanics on a
pseudo-Riemannian manifold $(X,g)$ can be formulated as follows.

The geodesic equation
\mar{cqg2}\beq
\ddot x^\m= K_\la^\m(x^\nu,\dot x^\nu) \dot x^\la, \label{cqg2}
\eeq
on the tangent bundle $TX$ with respect to a connection
\mar{cqg3}\beq
K=dx^\la\ot(\dr_\la +K^\m_\la\dot\dr_\m) \label{cqg3}
\eeq
on $TX\to X$ is called a relativistic geodesic equation if a
geodesic vector field of $K$ lives in the subbundle of
hyperboloids
\mar{cqg1}\beq
W_g=\{\dot x^\la\in TX\, \mid \,\,g_{\la\m} \dot x^\la\dot
x^\m=1\}\subset TX \label{cqg1}
\eeq
defined by the relativistic constraint (\ref{kk71}).

One can show that the equation (\ref{cqg2}) is a relativistic
geodesic equation if the condition
\mar{cqg4}\beq
(\dr_\la g_{\m\nu}\dot x^\m + 2g_{\m\nu}K^\m_\la)\dot x^\la \dot
x^\nu =0 \label{cqg4}
\eeq
holds.

Obviously, the connection (\ref{kk75}) fulfils the condition
(\ref{cqg4}). Any metric connection, e.g., the Levi--Civita
connection $\{_\la{}^\m{}_\nu\}$ on $TX$ satisfies the condition
(\ref{cqg4}).

Given a Levi--Civita connection $\{_\la{}^\m{}_\nu\}$, any
connection $K$ on $TX\to X$ can be written as
\mar{kk80}\beq
K^\m_\la = \{_\la{}^\m{}_\nu\}\dot x^\nu + \si^\m_\la(x^\la,\dot
x^\la), \label{kk80}
\eeq
where
\mar{x2}\beq
\si=\si^\m_\la dx^\la\ot\dot\dr_\la \label{x2}
\eeq
is some soldering form on $TX$. Then the condition (\ref{cqg4})
takes the form
\mar{cqg46}\beq
g_{\m\nu}\si^\m_\la\dot x^\la \dot x^\nu=0. \label{cqg46}
\eeq

With the decomposition (\ref{kk80}), one can think of the
relativistic geodesic equation (\ref{cqg2}):
\mar{kk81}\beq
\ddot x^\m= \{_\la{}^\m{}_\nu\}\dot x^\nu\dot x^\la +
\si^\m_\la(x^\la,\dot x^\la) \dot x^\la, \label{kk81}
\eeq
as describing a relativistic particle in the presence of a
gravitational field $g$ and a non-gravitational external force
$\si$.

\section{Hamiltonian relativistic mechanics}

We are in the case of relativistic mechanics on a
pseudo-Riemmanian world manifold $(X,g)$. Given the coordinate
chart (\ref{0301}) of its configuration space $X$, the homogeneous
Legendre bundle corresponding to the local non-relativistic system
on $U$ is the cotangent bundle $T^*U$ of $U$. This fact motivate
us to think of the cotangent bundle $T^*X$ as being the phase
space of relativistic mechanics on $X$. It is provided with the
canonical symplectic form
\mar{q25}\beq
\Om=dp_\la\w dx^\la \label{q25}
\eeq
and the corresponding Poisson bracket $\{,\}$.

A relativistic Hamiltonian is defined as follows
\cite{book98,rov91,sard98}. Let $H$ be a smooth real function on
$T^*X$ such that the morphism
\mar{gm616}\beq
\wt H: T^*X\to TX, \qquad \dot x^\m\circ \wt H=\dr^\m H,
\label{gm616}
\eeq
is a bundle isomorphism. Then the inverse image
\be
N=\wt H^{-1}(W_g)
\ee
of the subbundle of hyperboloids $W_g$ (\ref{cqg1}) is a
one-codimensional (consequently, coisotropic) closed imbedded
subbundle $N$ of $T^*X$ given by the condition
\mar{qq90'}\beq
H_T=g_{\m\nu}\dr^\m H\dr^\nu H-1=0. \label{qq90'}
\eeq
We say that $H$ is a relativistic Hamiltonian if the Poisson
bracket $\{H,H_T\}$ vanishes on $N$. This means that the
Hamiltonian vector field
\mar{rq11}\beq
\g=\dr^\la H\dr_\la -\dr_\la H\dr^\la \label{rq11}
\eeq
of $H$ preserves the constraint $N$ and, restricted to $N$, it
obeys the equation
\mar{gm610}\beq
\g\rfloor \Om_N +i_N^* d H=0, \label{gm610}
\eeq
which is the Hamilton equation of a Dirac constrained system on
$N$ with a Hamiltonian $H$ \cite{book05}.

The morphism (\ref{gm616}) sends the vector field $\g$
(\ref{rq11}) onto the vector field
\be
\g_T=\dot x^\la\dr_\la + (\dr^\m H\dr^\la\dr_\m H-\dr_\m
H\dr^\la\dr^\m H)\dot\dr_\la
\ee
on $TX$. This vector field defines the autonomous second order
dynamic equation
\mar{q35}\beq
\ddot x^\la=\dr^\m H\dr^\la\dr_\m H-\dr_\m H\dr^\la\dr^\m H
\label{q35}
\eeq
on $X$ which preserves the subbundle of hyperboloids (\ref{cqg1}),
i.e., it is the autonomous relativistic equation (\ref{k34}).

\begin{example}  \label{w770} \mar{w770}
The following is a basic example of relativistic Hamiltonian
mechanics. Given a one-form $A=A_\m dq^\m$ on $X$, let us put
\mar{kk101}\beq
H= g^{\m\n}(p_\m-A_\m)( p_\nu-A_\nu). \label{kk101}
\eeq
Then $H_T=2H-1$ and, hence, $\{H,H_T\}=0$. The constraint $H_T=0$
(\ref{qq90'}) defines a one-codimensional closed imbedded
subbundle $N$ of $T^*X$. The Hamilton equation (\ref{gm610}) takes
the form $\g\rfloor\Om_N=0$. Its solution (\ref{rq11}) reads
\be
&& \dot x^\al=g^{\al\nu}(p_\nu -A_\nu), \\
&& \dot p_\al=-\frac12\dr_\al g^{\m\nu}(p_\m-A_\m)( p_\nu-A_\nu)
+ g^{\m\nu}(p_\m-A_\m)\dr_\al A_\nu.
\ee
The corresponding autonomous second order dynamic equation
(\ref{q35}) on $X$ is
\mar{q12}\ben
&& \ddot x^\la- \{_\m{}^\la{}_\nu\}\dot x^\m\dot x^\nu -
g^{\la\nu}F_{\nu\m} \dot x^\mu=0, \label{q12}\\
&& \{_\m{}^\la{}_\nu\}=-\frac12g^{\la\bt}(\dr_\m
g_{\bt\nu}+\dr_\nu g_{\bt\m} -\dr_\bt g_{\m\nu}),
\nonumber\\
&& F_{\m\nu}=\dr_\m A_\nu-\dr_\nu A_\m. \nonumber
\een
It is a relativistic geodesic equation with respect to the affine
connection (\ref{kk75}).
\end{example}

Since the equation (\ref{q12}) coincides with the generic Lagrange
equation (\ref{kk72}) on a world manifold $X$, one can think of
$H$ (\ref{kk101}) as being a generic Hamiltonian of relativistic
mechanics on $X$.


\begin{thebibliography}{ederf}

\bibitem{eche} Echeverr\'{\i}a Enr\'{\i}quez, A., Mu\~noz Lecanda, M. and
Rom\'an Roy, N. (1991). Geometrical setting of time-dependent
regular systems. Alternative models, {\it Rev. Math. Phys.} {\bf
3}, 301.

\bibitem{book}  Giachetta, G., Mangiarotti, L. and Sardanashvily, G. (1997).
{\it New Lagrangian and Hamiltonian Methods in Field Theory}
(World Scientific, Singapore).

\bibitem{book05} Giachetta, G., Mangiarotti, L. and Sardanashvily, G.
(2005). {\it Geometric and Topological Algebraic Methods in
Quantum Mechanics} (World Scientific, Singapore).

\bibitem{epr2} Giachetta, G., Mangiarotti, L. and Sardanashvily, G.
(2006). Lagrangian and Hamiltonian dynamics of submanifolds,{\it
arXiv:} {\bf math-ph/0604066}.

\bibitem{book09} Giachetta, G., Mangiarotti, L. and Sardanashvily, G.
(2009). {\it Advanced Classical Field Theory} (World Scientific,
Singapore).

\bibitem{epr} Giachetta, G., Mangiarotti, L. and Sardanashvily, G.
(2009). Advanced mechanics. Mathematical introduction,{\it arXiv:}
{\bf 0911.0411}.

\bibitem{kras} Krasil'shchik, I., Lychagin, V. and Vinogradov, A. (1985). {\it Geometry of
Jet Spaces and Nonlinear Partial Differential Equations} (Gordon
and Breach, Glasgow).

\bibitem{leon} De Le\'on, M. and Rodrigues, P. (1989). {\it Methods of Differential
Geometry in Analytical Mechanics} (North-Holland, Amsterdam).

\bibitem{book98} Mangiarotti, L. and Sardanashvily, G. (1998). {\it Gauge
Mechanics} (World Scientific, Singapore).

\bibitem{rov91} Rovelli, C. (1991). Time in quantum gravity: A hypothesis, {\it Phys.
Rev.} {\bf D43}, 442.


\bibitem{sard98} Sardanashvily, G. (1998). Hamiltonian time-dependent mechanics,
{\it J. Math. Phys.} {\bf 39}, 2714.


\end{thebibliography}
\end{document}